\documentclass[letterpaper, 10 pt, conference]{ieeeconf}  

\IEEEoverridecommandlockouts                              

\usepackage{cite}
\usepackage{amsmath,amssymb,amsfonts}
\usepackage{algorithmic}
\usepackage{graphicx}
\usepackage{subcaption}
\usepackage{pgfplots}
\usepackage{tikz}
\pgfplotsset{compat=newest} 

\newtheorem{problem}{Problem}
\newtheorem{lemma}{Lemma}

\newtheorem{prop}{Proposition}
\newlength\fwidth
\newlength\fheight
\title{\LARGE \bf
Sparse Resource Allocation for Control of Spreading Processes via Convex Optimization
}

\author{Vera L. J. Somers and Ian R. Manchester
\thanks{The authors are with the Australian Centre for Field Robotics (ACFR), School of Aerospace, Mechanical and Mechatronic Engineering,
        University of Sydney, NSW 2006, Australia
        {\tt\small \{v.somers, i.manchester\}@acfr.usyd.edu.au}}%
}

\begin{document}

\maketitle
\thispagestyle{empty}
\pagestyle{empty}

\begin{abstract}
In this letter we propose a method for sparse allocation of resources to control spreading processes -- such as epidemics and wildfires -- using convex optimization, in particular exponential cone programming. Sparsity of allocation has advantages in situations where resources cannot easily be distributed over a large area. In addition, we introduce a model of risk to optimize the product of the likelihood and the future impact of an outbreak. We demonstrate with a simplified wildfire example that our method can provide more targeted resource allocation compared to previous approaches based on geometric programming.
\end{abstract}

\section{INTRODUCTION}
Contagious diseases, computer viruses, and wildfires can all be thought of as spreading processes in which an initial localized outbreak spreads rapidly to neighboring nodes in a network \cite{Karafyllidis1997a,Bloem2008,Nowzari2016}. Because of the real-world risks associated with such events, there has been significant research into methods for modeling, prediction, and control. 

Spreading processes typically evolve over very large networks, e.g. global travel networks for epidemics, large geographic areas for wildfires and the internet for computer viruses. Therefore scalability of computational methods is important. Furthermore, sparsity of resource allocation solutions is often needed, because it can be difficult to distribute resources broadly. 
  
Spreading processes are commonly modeled as Markov processes. The most well-known models are the Susceptible-Infected-Susceptible (SIS) model and the Susceptible-Infected-Removed (SIR) model \cite{kermark1927contributions,bailey1975mathematical}. These stochastic models can be approximated as ordinary differential equation (ODE) models, which can in turn be approximated by a linear model \cite{Ahn2013,VanMieghem:2009,Nowzari2016} which is proven by \cite{VanMieghem:2009} to be an upperbound and is therefore usually the object of study.  

The problems of minimizing the spreading rate by removing either a fixed number of links or nodes in the network are both NP-hard \cite{van2011decreasing}. This fact has motivated the study of heuristics methods based on node-rankings of various forms \cite{Hadjichrysanthou2015,Dhingra2018,Lindmark}. However, in general these approaches will not be optimal in any sense and furthermore the assumption of complete link or node removal is often unrealistic.

A more realistic assumption is that spreading rate can be decreased and the recovery rate increased. This can be achieved applying resources to the nodes and links. Various methods have been proposed where the resource allocation is subject to budget constraints (e.g. \cite{Nowzari2016,Giamberardino2017,Bloem2008,Khanafer,DiGiamberardino2019,Mai2018,Liu2019b,Dangerfield2019,Torres2017,Preciado2014,Zhang2018,Han2015,Nowzari2017}). 

However, most of these do not result in sparse resource allocation. 
In addition, most of these papers consider minimizing the dominant eigenvalue of the linear dynamics, i.e. the overall spreading rate across the network. However, in many cases it is important to take into account node-dependent costs. For example, higher cost may associated with nodes representing populated areas when controlling a wildfire, or  more vulnerable members of the community in an epidemic.

A multi epidemic problem is considered in \cite{Dangerfield2019} and treated as a knapsack problem and sparse resource allocation is obtained. Another method for sparse resource allocation for linear network spread dynamics is proposed in \cite{Torres2017}, and global optimality is proven in the special case of diagonally symmetrizable matrices. 

In general, the problem of designing  sparse feedback gains for linear systems is non-convex and computationally challenging \cite{Lin2013}. However, linear spreading processes are positive systems \cite{berman1994nonnegative}, which enables control based on linear programming (LP) \cite{rantzer2015scalable} and geometric programming (GP) \cite{Boyd2007}, which allows global optima to be found with efficient numerical methods.

Optimal resource allocation via geometric programming has been studied in \cite{Preciado2014,Zhang2018,Han2015,Nowzari2017, Ogura2019}. The work most similar to ours is \cite{Preciado2014}, but our approach differs in both the cost function associated with spread and the resource model. 

In particular, the contributions of this paper are two-fold: firstly, a risk model that based on the product of the likelihood of an outbreak with its discounted future cost, with node-dependent weightings. Secondly, we propose a resource model which leads to sparse resource allocation. Taken together, we show via a simplified wildfire example that our approach can lead to more precisely targeted allocation of resources. Our proposed formulation is not technically a GP, but is similar in the sense that it is convex under logarithmic transformation, in particular it is an exponential cone program. Furthermore, under this transformation our resource model corresponds to $\ell_{1}$ type constraints which are known to encourage sparsity \cite{tibshirani1996regression, candes2006robust, donoho2006compressed} and sparsity can be further increased using reweighted iterations \cite{Candes2008}.

\section{PROBLEM AND MODEL FORMULATION}

\subsection{Notations}
The Hadamard product, i.e. element-wise multiplication, is indicated with the $\odot$ notation. $A\geq B$ indicates all elements $a_{ij}\geq b_{ij}$, in particular $A \geq 0$ indicates all elements are non-negative. All other notation is standard.


\subsection{SIS Spreading Process Model}
We study a spreading process on a graph $\mathcal{G}$ with $n$ nodes and edge set $\mathcal{E}$, where each node $i \in \{1, 2, ..., n\}$ has a state $X_i(t)$ associated with it. For the basic SIS model \cite{kermark1927contributions} a node can be in two states: infected, i.e. $X_i(t)=1$, or susceptible to infection from neighboring nodes, i.e. $X_i(t)=0$. An infected node recovers with probability $\delta_{i}\Delta t$ to $X_i(t+\Delta t)=0$ and the process spreads from infected node $j$ to susceptible node $i$ with probability $\beta_{ij} \Delta t$. 

We now define $x_{i}(t)=E(X_{i}(t))=P(X_{i}(t)=1)$ as the probability of a node $i$ being infected at time $t$. Using a mean-field estimation and Kolmogorov forward equations to build an approximate deterministic model from the stochastic model \cite{Nowzari2016,Preciado2014}, we obtain $n$ coupled nonlinear Markov differential equations
\begin{equation}
\label{eq:nonL}
\dot{x}_{i}(t)=(1-x_{i}(t))\sum^{n}_{j=1}\beta_{ij}x_{j}(t)-\delta_{i}x_{i}(t).
\end{equation}
Here, the main assumption taken is that all pairs of random variables have zero covariance.

A linear model is now obtained by linearizing this deterministic model around the infection-free equilibrium point ($\mathbf{x}=0$) \cite{Preciado2014} and we can define the system on the graph by the linear differential equation
\begin{equation}
\label{eq1}
   \dot{x}(t)=Ax(t) 
\end{equation}
where $x(t)=[x_{1}(t),...,x_{n}(t)]^{T}$ with $t\geq0$, is the state of the system and the sparse state matrix $A$ is defined by
 \begin{equation}
  \label{eq:epi}
 a_{ij} = 
\begin{cases}
-\delta_{i}  \le 0 &\quad \text{if}\quad  i=j, \\
\beta_{ij} \ge 0&\quad \text{if}\quad  i\neq j, (i,j) \in \mathcal{E}, \\
0  &\quad \text{otherwise.}
 \end {cases}
 \end{equation}
Because all off-diagonal entries $a_{ij}$ are assumed to be nonnegative, $A$ is Metzler and the system is positive, i.e. if $x_i(0)\ge 0$ for all $i$, then $x_i(t)\ge 0$ for all $t\ge 0$ \cite{berman1994nonnegative}. 

It is proven in \cite{Mieghem2011,Li2012} that the obtained probabilities from the approximation (\ref{eq:nonL}) upper bound the true values. Furthermore the linear model also upper bounds the nonlinear model \cite{VanMieghem:2009,Preciado2014}. Therefore similar to many papers, we study the linear model to control the underlying process. The linear model accurately captures the initial exponential phase of growth in which intervention is essential, however becomes less accurate as a large percentage of nodes is affected. 

\subsection{Risk Model}
\label{subsec:R}
To construct our model of risk, we first define the following cost function:
\begin{equation}
\label{eq:cost2}
J\left(x(0)\right)=  \int_{0}^\infty e^{-rt} Cx(t)dt
\end{equation}
with the system where $x(t)$ satisfies (\ref{eq1}) and $C= [c_1, ..., c_n]$ is a row vector defining the cost associated with each node $i$, with each $c_i\ge 0$. The discount rate $r>0$ can be tuned to emphasize near-term cost over long-term cost. If $r$ is large enough such that $A-rI$ is Hurwitz-stable, i.e. all eigenvalues have negative real parts, then $J\left(x(0)\right)$ is finite for all $x(0)$ and a linear function of the initial state, i.e.
\begin{equation}
\label{eq:R2}
   \int_{0}^\infty e^{-rt} Cx(t)dt = \sum_{i=1}^n p_i x_i(0)
\end{equation}
where $p_{i}x_i(0)$ can be seen as the discounted cost-to-go associated with each node $i$. That is, if a spreading process would start at node $i$ what will the future discounted cost be of the process spreading from there over the graph $\mathcal{G}$. We, therefore, define the vector $p$ as the \textit{node impact}, which can also be interpreted as a node priority for surveillance of spreading processes as was proposed in \cite{icrapaper2019}.

We can now find the node impact $p$ via two different methods. First, by direct calculation
\begin{equation}
\label{eq:MM}
p^{T}=C\left (rI-A\right )^{-1}
\end{equation}
as derived in \cite{icrapaper2019}. From this representation we can establish two useful properties:
\begin{lemma}
\label{L2}
Each element of $p$ is non-negative and a monotone function of each element of $A$.
\end{lemma}

\begin{proof}
 $(rI-A)$ is a positive-stable M-matrix, which is inverse positive \cite[p.~134]{berman1994nonnegative}, i.e. all elements of the matrix $\left (rI-A\right )^{-1}$ are non-negative, as are all elements of $C$ by construction. Therefore $p_{i}$ will always be non-negative. Furthermore $(rI-A)$ is non-singular and given two non-singular M-matrices $A$ and $B$, if $A\geq B$, then $B^{-1}\geq A^{-1}$ \cite{fiedler1962matrices}. When we combine this with the fact that all elements of $C$ are non-negative, this implies that reducing spreading rate $\beta_{ij}$, i.e. reducing $A$, hence making $rI-A$ less negative element wise, i.e. making it larger element wise and hence reducing the inverse, will reduce the node impact $p$ as given in (\ref{eq:MM}), i.e. $p_{i}$ is a monotone function of $A_{jk}$.
\end{proof}

Lemma \ref{L2} implies that reducing the spreading rate or increasing recovery rate can never increase the node impact and vice-versa. An important application of this is robust solutions: if spreading or recovery rates are uncertain, but known to be in an interval, then worst-case node impact can be calculated using the boundary values of the intervals.

The second method to find the node impact $p$ is via a linear program (LP) which is suitable for extension to include resource allocation. The equivalent LP is
\begin{equation}
\label{eq:LP1}
\begin{split}
    \text{minimize} & \quad \quad  |p|_{1} \\
    \text{such that} & \quad \quad p \geq 0, \quad \quad p^{T}A - r p^{T} \leq - C
\end{split}
\end{equation}

\begin{lemma}
The LP (\ref{eq:LP1}) is equivalent to (\ref{eq:MM}).
\end{lemma}
\begin{proof}
Equivalence can be shown by application of Lemma \ref{L2} and the properties of LPs. Let $\bar{p}$ be the node impact calculated using (\ref{eq:MM}), then $\bar p\ge 0$ and $\bar{p}^{T}A - r \bar{p}^{T} =- C $ hence $\bar{p}$ is feasible for the LP, but any other feasible $p$ has $p^{T}A - r p^{T} \le - C$ or $p^T\ge =C(rI-A)^{-1} = \bar{p}$, so $|p|_{1}=\sum p_{i} \geq \sum \bar{p}_{i} = |\bar{p}_{1}|$and $p=\bar{p}$ is optimal for the LP.
\end{proof}

We can now define a risk model associated with the spreading process on the graph $\mathcal{G}$. We define the risk $R_{i}$ at node $i$ as the product of the likelihood of an outbreak starting at node $i$, i.e. the estimated probability $\hat{x}_{i}(0)$, and the node impact $p_{i}$. The bounded risk of an outbreak occurring can now be defined as
\begin{equation}
\label{eq:Risk2}
   R=\hat{x}(0) \odot p.
\end{equation}

In this letter we focus on minimizing the maximum risk $\max_{i}(R_{i})$, i.e. we allocate resources to reduce the impact of the worst expected localized outbreak. However the total risk, i.e. $\sum_{i} R_{i}$, can also be taken. We compare our risk model with the cost function in \cite{Preciado2014} and others:
\begin{equation}
\label{eq:eigP}
   J_{\lambda}=\lambda_{max}(A)
\end{equation}
where $\lambda_{max}$ is the dominant eigenvalue, i.e. the eigenvalue with largest real part. 

\subsection{Resource Allocation Model}
Now that the spreading process model and risk associated with it are defined, we can look into how to allocate resources to the system. We aim to reduce the risk by reducing the spreading rate and increasing recovery rate within defined bounds. That is the updated $\beta_{ij}$ and $\delta_{i}$  are restricted by respectively $0 < \underline{\beta}_{ij} \leq \beta_{ij} \leq \overline{\beta}_{ij}$ and $0 < \underline{\delta}_{i} \leq \delta_{i} \leq \overline{\delta}_{i}<1$. We now propose to define the resource allocation models as 
\begin{equation}
\label{eq:RM}
f_{ij}\left(\beta_{ij}\right)=w_{ij}\text{log}\left(\frac{\overline{\beta}_{ij}}{\beta_{ij}}\right), \quad g_{i}\left(\delta_{i}\right)=w_{ii}\text{log}\left(\frac{1-\underline{\delta}_{i}}{1-\delta_{i}}\right)
\end{equation}
where $w_{ij}$ and $w_{ii}$ are weightings that indicate the cost of respectively reduction of spreading rate $\beta_{ij}$ and increase of recovery rate $\delta_{i}$. E.g. $w_{ij}\log(2)$ is the cost of reducing the spread rate from $i$ to $j$ to half its original value.

These logarithmic resource models can be understood as a proportional decrease. That is, a reduction in $\beta$ (and increase in $\delta$) in certain proportion always takes the same amount of resources, because $\text{d}f_{ij}=-\frac{\text{d}\beta_{ij}}{\beta_{ij}}$. Furthermore this implies that it is impossible for $\beta_{ij}$ to become $0$. This corresponds with real spreading scenarios where it is impossible to completely eliminate the possibility of spread. 

We compare our proposed resource model with the resource model in \cite{Preciado2014}, i.e.
\begin{equation}
\label{eq:Preciado}
f_{ij}\left(\beta_{ij}\right)=\frac{\beta^{-1}_{ij}-\overline{\beta}^{-1}_{ij}}{\underline{\beta}^{-1}_{ij}-\overline{\beta}^{-1}_{ij}}, \quad g_{i}\left(\delta_{i}\right)=\frac{(1-\delta_{i})^{-1}-(1-\underline{\delta}_{i})^{-1}}{(1-\overline{\delta}_{i})^{-1}-(1-\underline{\delta}_{i})^{-1}}
\end{equation}
as visualized for the spreading rate in Fig. \ref{fig:RModel} for $\underline{\beta}=0.05$ and $\overline{\beta}=1$. Notice that the  proposed model associates higher cost for low resource investments, encouraging sparse allocation.

%
\begin{figure}
  \centering
 \def\svgwidth{0.4\textwidth}
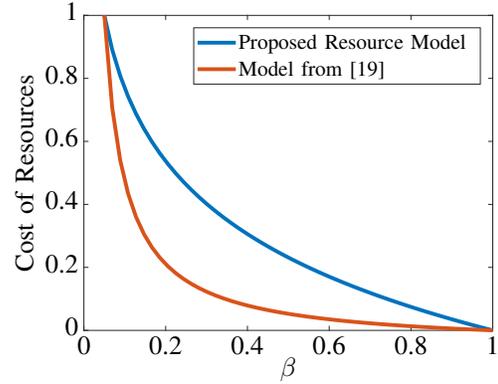
    \caption{Comparison of resource models (\ref{eq:RM}) and (\ref{eq:Preciado}) for investment on spreading rate $\beta$}
    \label{fig:RModel}
\end{figure}
%


\subsection{Problem Statements}
We want to keep both the risk and the allocation of resources small. Therefore we study two closely related problems of sparse resource allocation for spreading processes: 

\begin{problem}[Resource-Constrained Risk Minimization]
Given a defined resource allocation budget $\Gamma$, a cost $c_{i}$ associated with each node $i$, find the optimal spreading and recovery rates $\beta_{ij}$ and $\delta_{i}$ that via sparse resource allocation minimize the maximum risk $\hat{x}(0)\odot p$, i.e. find the updated state matrix $A$ that minimizes
\begin{align}
    \underset{p, \beta, \delta}{\text{minimize}} & \quad \quad  \text{max}(\hat{x}(0)\odot p) \label{eq:C0}\\
    \text{such that} & \quad \quad p^{T}A - r p^{T} \leq - C \label{eq:C1}\\
 & \quad \quad \sum_{ij} f_{ij}\left(\beta_{ij}\right) + \sum_{i} g_{i}\left(\delta_{i}\right) \leq \Gamma \label{eq:C2}\\ 
& \quad \quad p \geq 0, \quad 0 < \underline{\beta}_{ij} \leq \beta_{ij} \leq \overline{\beta}_{ij} \label{eq:C3}\\
& \quad \quad 0 < \underline{\delta}_{i} \leq \delta_{i} \leq \overline{\delta}_{i} < 1\label{eq:C33}
\end{align} 
where $A$ is defined as per (\ref{eq:epi}).
\end{problem}

\begin{problem}[Risk-Constrained Resource Minimization]
Find the optimal spreading and recovery rates $\beta_{ij}$ and $\delta_{i}$ that via sparse resource allocation minimize the amount of resources required, given an upper bound on the maximum risk $\gamma$ and a cost $c_{i}$ associated with each node $i$, i.e. find the updated state matrix $A$ that minimizes
\begin{align}
     \underset{p, \beta, \delta}{\text{minimize}} & \quad \quad   \sum_{ij} f_{ij}\left(\beta_{ij}\right) + \sum_{i} g_{i}\left(\delta_{i}\right) \\
    \text{such that} & \quad \quad p^{T}A - r p^{T} \leq - C \\
    & \quad \quad p_{i}\hat{x}_{i}(0) \leq \gamma,  \quad p \geq 0\\
& \quad \quad 0 < \underline{\beta}_{ij} \leq \beta_{ij} \leq \overline{\beta}_{ij} \\
& \quad \quad 0 < \underline{\delta}_{i} \leq \delta_{i} \leq \overline{\delta}_{i} < 1
\end{align}
where $A$ is defined as per (\ref{eq:epi}).
\end{problem}

\section{A CONVEX FRAMEWORK FOR SPARSE RESOURCE ALLOCATION }
In this section we show that Problems 1 and 2 can be reformulated as convex optimization problems, in particular exponential cone programs, which recent versions of commercially available solvers, e.g. MOSEK, can solve efficiently. Furthermore, we discuss how the proposed resource model leads to sparse resource allocation. Our problem formulations are not technically GPs, but are similar in that they are convex after logarithmic transformation.

To save space we present here the constraints needed in both the optimization problems that we formulate below: 
\begin{align}
&  \text{log} \left (\sum_{i \neq j} \text{exp} \left(y_{i} + \text{log}\left(\frac{\overline{\beta}_{ij}}{1+r}\right) - u_{ij} -  y_{j} \right) \right. \nonumber\\ 
&  + \text{exp}\left(\text{log}\left(\frac{1-\underline{\delta}}{1+r}\right) - v_{j}\right) \nonumber \\
&  \left.+ \text{exp}\left (\text{log}\left (\frac{C_{j}}{1+r}\right) - y_{j}\right)\right) \leq 0 \quad \forall j  \label{eq:C4},
\end{align}
\begin{equation}
0 \leq u_{ij} \leq w_{ij}\text{log}\left (\frac{\overline{\beta}_{ij}}{ \underline{\beta}_{ij}}\right ) \label{eq:C6},
\end{equation}
\begin{equation}
0 \leq v_{i} \leq w_{ii}\text{log}\left (\frac{1-\underline{\delta}_{i}}{ 1-\overline{\delta}_{i}}\right ) \label{eq:C7}. 
\end{equation}
 
\begin{prop} 
Problem 1 is equivalent to the following convex optimization problem under the transformation $y=\text{log}(p)$ and $u_{ij}=f_{ij}\left(\beta_{ij}\right)$ and $v_{i}=g_{i}\left(1- \delta_{i}\right)$ 
\begin{align}
     \underset{y, u, v}{\text{minimize}} & \quad \quad  \text{max}(\text{log}(\hat{x}(0))+ y) \label{eq:OF}\\
    \text{such that} & \quad \quad (\ref{eq:C4}), (\ref{eq:C6}), (\ref{eq:C7}), \nonumber \\
& \quad \quad  \sum_{ij} u_{ij} + \sum_{i} v_{i} \leq \Gamma \label{eq:C5}.
\end{align}
\end{prop}

\begin{proof}
The objective function (\ref{eq:OF}) follows directly from (\ref{eq:C0}) and $y=\text{log}(p)$. To obtain constraint (\ref{eq:C4}) from (\ref{eq:C1}) we take that (\ref{eq:C1}) is equivalent to $\sum^{n}_{i=1}p_{i}\left (A_{ij}-rI\right ) \leq -C_{j}$ for all $j$. Now using (\ref{eq:epi}) this can be rewritten as 
\begin{equation}
\sum_{i\neq j}p_{i}\beta_{ij}-p_{j}\delta_{j}-p_{j}r \leq - C_{j}, \quad \forall j
\end{equation}
which is equivalent to
\begin{equation}
\label{eq:posyC1}
\sum_{i \neq j} \frac{p_{i}\beta_{ij}}{p_{j}(1+r)} + \frac{1-\delta_{j}}{1+r}+ \frac{c_{j}}{p_{j}(1+r)} \leq 1 \quad \forall j.
\end{equation}
Taking the log of both sides
and rewriting gives (\ref{eq:C4}). Now, (\ref{eq:C5}) follows directly from (\ref{eq:C2}) and $u_{ij}=f_{ij}\left(\beta_{ij}\right)$ and $v_{i}=g_{i}\left(1-\delta_{i}\right)$. Finally rewriting (\ref{eq:C3}) gives $ 0 < \frac{\underline{\beta}_{ij}}{\overline{\beta}_{ij}} \leq \frac{\beta_{ij}}{\overline{\beta}_{ij}} \leq 1 $ which is equivalent to $ 0 \leq \text{log}\left(\frac{\overline{\beta}_{ij}}{\beta_{ij}} \right) \leq \text{log}\left(\frac{\overline{\beta}_{ij}}{\underline{\beta}_{ij}} \right) $ and can be rewritten to (\ref{eq:C6}) using $u_{ij}=w_{ij}\text{log}\left(\frac{\overline{\beta}_{ij}}{\beta_{ij}}\right)$. The bounds on $v_{i}$ (\ref{eq:C7}) can be found in the same way. To show that this optimization problem is convex, we can use the fact that monomials and posynomials are convex in log-scale \cite{boyd2004convex}. The objective and all constraints except (\ref{eq:C1}) of Problem 1 are already defined as such. Using the rewritten constraint (\ref{eq:posyC1}), we obtain a posynomial constraint and hence, our optimization problem is convex in log scale. 
\end{proof}

\begin{prop}
Problem 2 is equivalent to the following convex optimization problem under the transformation $y=\text{log}(p)$ and $u_{ij}=f_{ij}\left(\beta_{ij}\right)$ and $v_{i}=g_{i}\left(1- \delta_{i}\right)$
\begin{align}
    \underset{y, u, v}{ \text{minimize}} & \quad \quad \sum_{ij} u_{ij} + \sum_{i} v_{i}  \label{eq:l1}\\
    \text{such that} & \quad \quad  (\ref{eq:C4}), (\ref{eq:C6}), (\ref{eq:C7}), \nonumber \\
& \quad \quad \text{max}(\text{log}(\hat{x}(0))+ y) \leq \text{log}(\gamma). 
\end{align}
\end{prop}

\begin{proof}
Similar to the proof of Proposition 1.
\end{proof}

The convex optimization formulation also clearly shows why our proposed resource model encourages sparsity. Our resource model, now formulated as constraint (\ref{eq:C5}) and objective (\ref{eq:l1}) are an $\ell_{1}$ type constraint and objective that are known to encourage sparsity \cite{tibshirani1996regression,candes2006robust,donoho2006compressed,Candes2008}.

\subsection{Reweighted $\ell_{1}$ minimization}
\label{subsec:L0}
If the goal is maximal sparsity, i.e. minimal number of nodes with non-zero resources allocation, then we can apply the reweighted $\ell_{1}$ optimization approach of \cite{Candes2008}.  We can apply this to our problem by iteratively solving Problem 1 or 2, but with a reweighted resource model that approximates the number of nodes with non-zero allocation:
\begin{equation}
\label{eq:L1}
h^k=\sum_{ij} \frac{u^{k}_{ij}}{u_{ij}^{k-1}+\epsilon}+\sum_{i} \frac{v^{k}_{i}}{v_{i}^{k-1}+\epsilon}
\end{equation} 
where $k$ is the iteration number and $\epsilon$ a very small number to improve numerical stability. For Problem 1 we now replace constraint (\ref{eq:C5}) with $h^k \leq M$ where $M$ is the bound on the number of nodes and links that can have resources allocated to them. For Problem 2 the objective changes from minimizing (\ref{eq:l1}) to minimizing $h^k$. This iteration has no guarantee of convergence or global optimality, but has been found to be very effective in many cases.

\section{NUMERICAL RESULTS}
We illustrate our method with a simplified model of a wildfire. Let us consider the fictional landscape given in Fig. \ref{fig:Landscape} consisting of different vegetation types, a city and water. We represent this landscape as a network graph with $n=1000$ nodes, where the set of edges $\mathcal{E}$ is based on an 8-node spreading direction grid, i.e. fire can spread from each node to its direct neighboring 8 nodes (horizontal, vertical and diagonal). 

The spreading rates are now determined by the vegetation type, wind speed and direction. These values are computed based on data from wildfire models given in \cite{Karafyllidis1997a} and \cite{Alexandridis2008a}, where the stochastic spreading rates are of the form
\begin{equation}
\beta=\beta_{b}\beta_{veg}\beta_{w}.
\end{equation}
The baseline spreading rate $\beta_{b}=0.5$ and $ \beta_{veg}=0.1, 1$ and $1.4$ for respectively desert, grassland and eucalyptic forest. Water is considered unburnable and those edges are removed, resulting in a total number of $3486$ non-zero edges. $\beta_{w}$ is calculated following \cite{Alexandridis2008a} for a northeasterly wind of $V=4$ m/s. Furthermore the spreading rate $\beta$ is corrected for spreading between diagonally connected nodes, following \cite{Karafyllidis1997a}. The cost of the city nodes is taken as $c_{i}=1$, whereas $c_{i}=0.01$ for all other nodes. Finally, the discount rate is set to $r=3.5$ and we take into account a fire likelihood map as depicted in Fig. \ref{fig:LM}. For simplicity we will only consider resource allocation on the spreading rate $\beta$ and take $\delta=0.2$ and $w_{ij}=1$ for all edges $(i,j)$.

\begin{figure}
  \centering
    \includegraphics[width=0.95\linewidth]{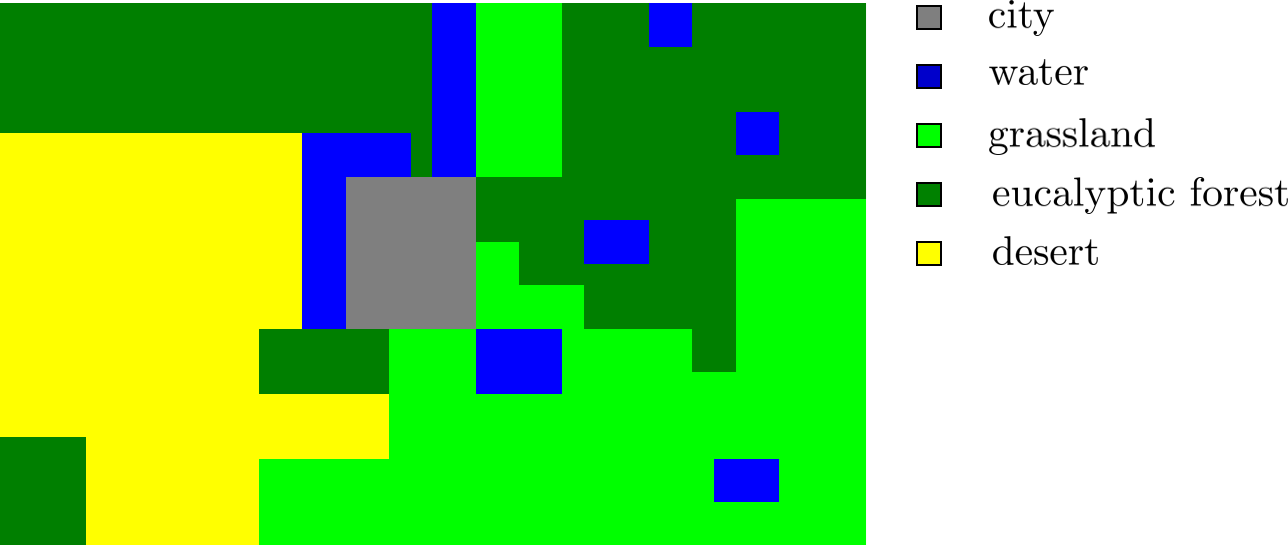}
    \caption{Fictional landscape with different area types, represented as a grid with $n=1000$ nodes.}
    \label{fig:Landscape}
\end{figure}

\begin{figure}
  \centering
   \def\svgwidth{0.4\textwidth}
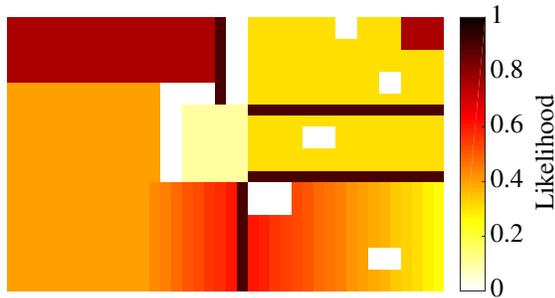
    \caption{Likelihood Map representing the likelihood of a fire outbreak in landscape Fig. \ref{fig:Landscape}.}
    \label{fig:LM}
\end{figure}

The optimization problems are solved with YALMIP \cite{lofberg2004yalmip} and MOSEK in Matlab, and all can be solved on a standard desktop computer within seconds.\footnote{Code available on https://github.com/imanchester/SpreadingProcesses.} 

We compare both our proposed  risk model and resource model with those presented in \cite{Preciado2014}. But to have a comparable results we must consider each in turn.

Firstly, we compare our proposed risk model (\ref{eq:Risk2}) with minimizing the dominant eigenvalue (\ref{eq:eigP}). We do this via Problem 1, the budget-constrained resource allocation: we fix a constraint the proposed resource model (\ref{eq:RM}) and compare minimizing (\ref{eq:Risk2}) to minimizing (\ref{eq:eigP}). We take a resource allocation budget of $\Gamma=25$ and $\underline{\beta_{i}}=1 \times 10^{-4}$ for all nodes. The resulting allocations are shown in Fig. \ref{fig:P1}. If the link is plotted that indicates there is a resources allocation to that edge, where the color indicates the ratio of reduction $u_{ij}$. Here red indicates full reduction to $\underline{\beta}_{i}$ and the darker blue the lower the investment on that edge. 

\begin{figure}
     \begin{subfigure}[b]{0.95\linewidth}
\centering
   \def\svgwidth{1\textwidth}
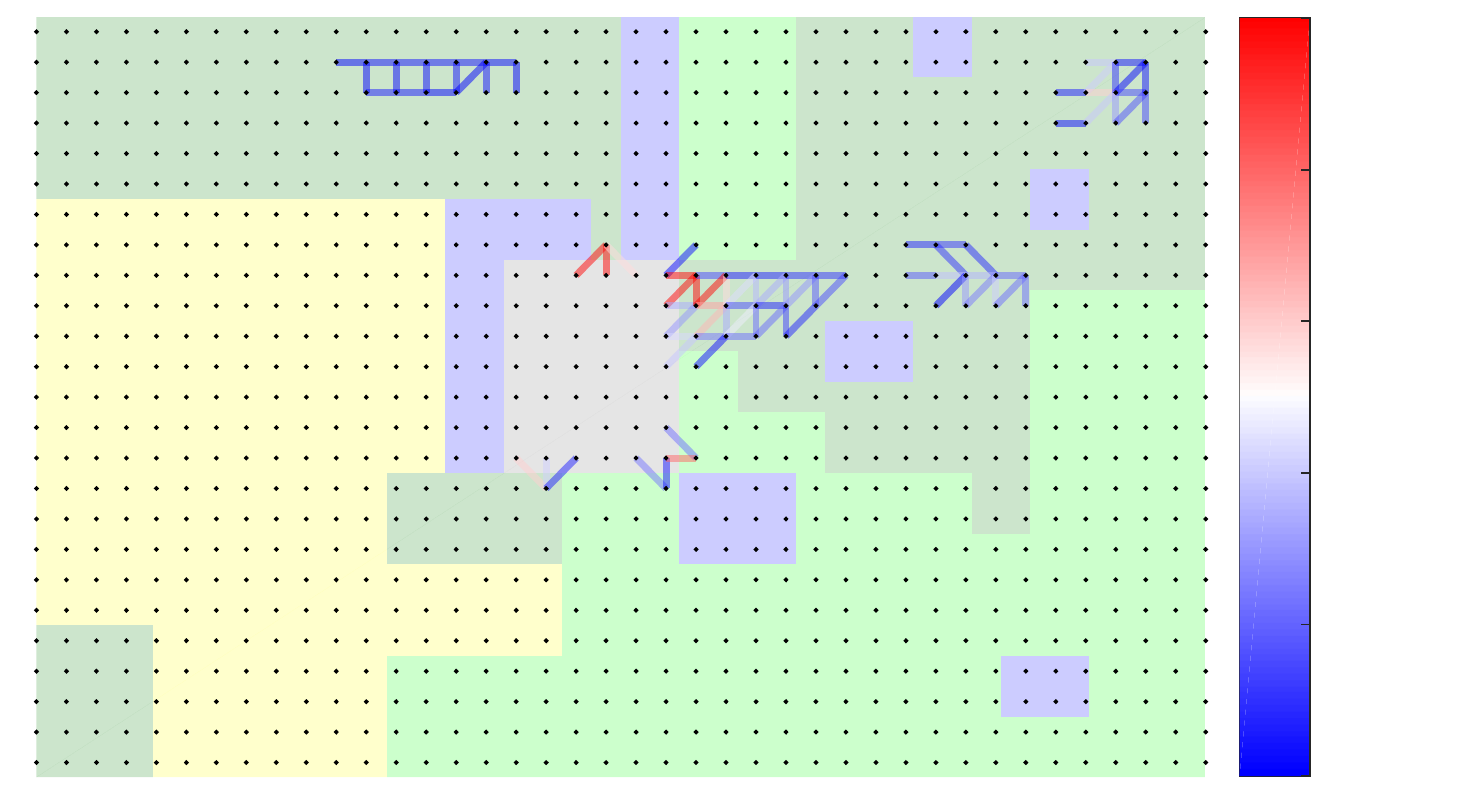
\caption{Minimize the proposed risk function (\ref{eq:Risk2}) subject to constraints on the proposed resource model (\ref{eq:RM}).}
      \label{fig:P1Ours}
  \end{subfigure}
  ~ 
\begin{subfigure}[b]{0.95\linewidth}
        \def\svgwidth{\textwidth}
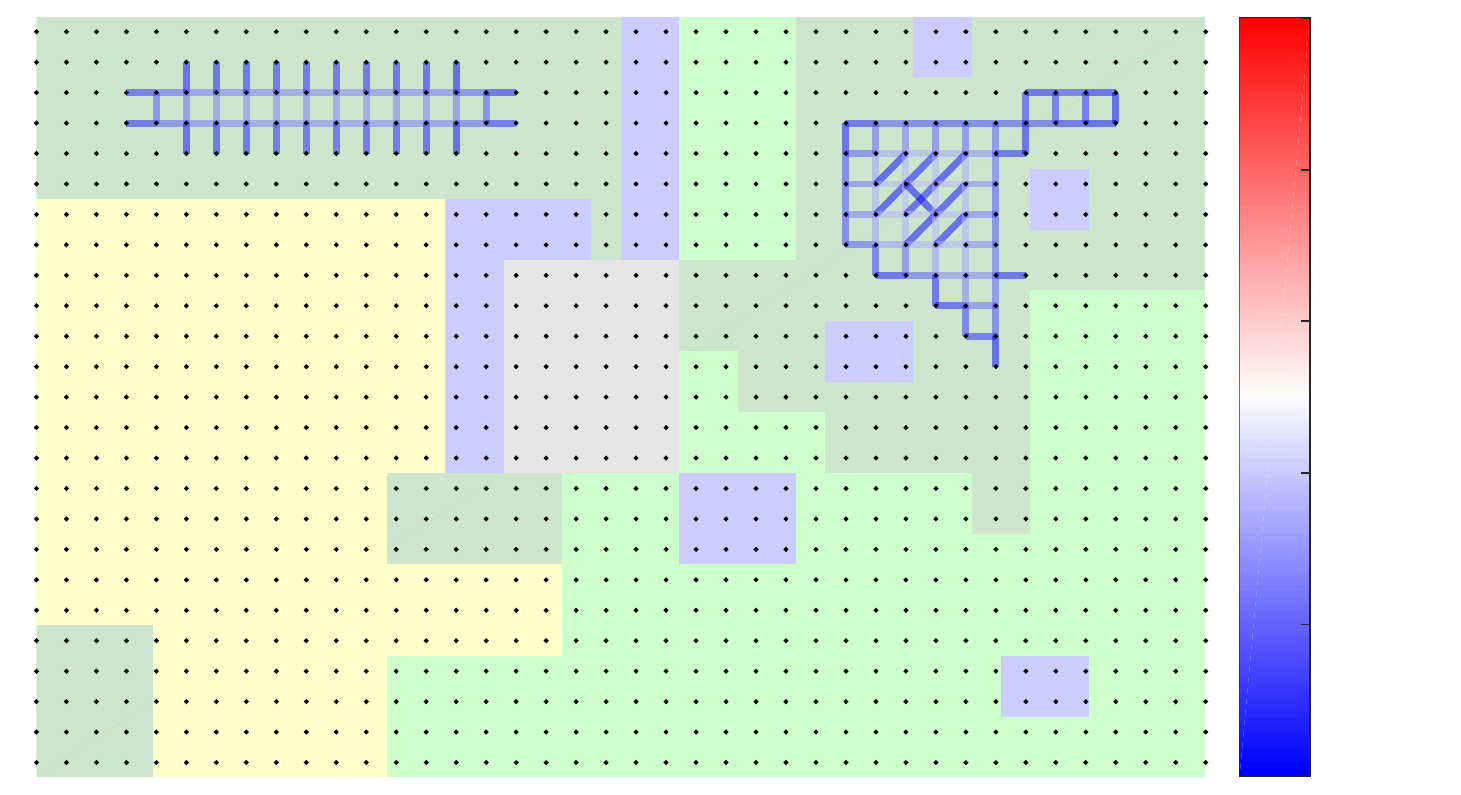
      \caption{Minimize the dominant eigenvalue (\ref{eq:eigP}) from \cite{Preciado2014} subject to constraints on the proposed resource model (\ref{eq:RM}).}
      \label{fig:P1Preciado}
  \end{subfigure}
      \caption{Resource Allocation Map for Problem 1.}
\label{fig:P1}
\end{figure}

Our approach allocates resources in such a way that the worst-case risk of any localised outbreak is minimized. In particular in Fig. \ref{fig:P1Ours} it can be seen that the model results in containment lines to protect high cost areas from areas with high risk of spread. If the dominant eigenvalue is minimized (Fig. \ref{fig:P1Preciado}), all areas of the landscape is considered equally important and containment lines are not obtained.  

Secondly, we compare resource models. We do this via Problem 2, i.e. risk-constrained resource minimization. The risk bound that we use is $\gamma=0.0516$, which was the risk bound achieved via Problem 1 above, as plotted in Fig. \ref{fig:P1Ours}. Therefore, Fig. \ref{fig:P1Ours} also shows the solution for Problem 2 minimizing our proposed resource model (\ref{eq:RM}) subject to this resource constraint. In Fig. \ref{fig:Preciado008} we show the results for minimizing resource model (\ref{eq:Preciado}) from \cite{Preciado2014} subject to this same resource constraint.

\begin{figure}
\centering
         \def\svgwidth{0.475\textwidth}
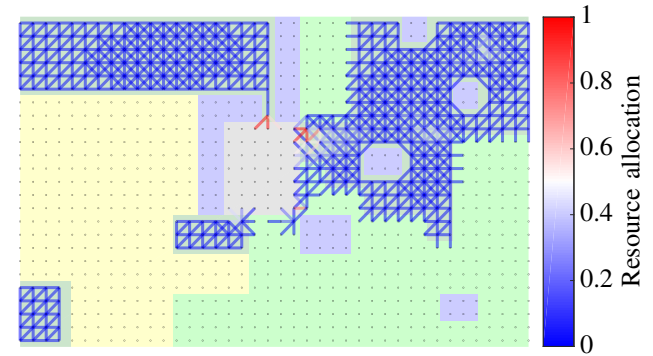
      \caption{Resource Allocation Map for Problem 2 for minimizing the resource model (\ref{eq:Preciado}) from \cite{Preciado2014} subject to constraints on the proposed risk model (\ref{eq:Risk2}).}
      \label{fig:Preciado008}
\end{figure}

We can observe that the resource model of \cite{Preciado2014} allocates a low investment on a large number of nodes. This is due to the low penalty for small investments, whereas our proposed method encourages more sparse allocation. Out of 3486 total edges, resource model (\ref{eq:Preciado}) from \cite{Preciado2014} allocates resources on 1109 edges (Fig. \ref{fig:Preciado008}), whereas our proposed method only invests on 89 edges (Fig. \ref{fig:P1Ours}). 

To further improve sparsity we solve Problem 2 with the reweighted $\ell_{1}$ minimization as explained in Section \ref{subsec:L0}. For Problem 2 we keep the constraints the same, but iteratively minimize (\ref{eq:L1}). The obtained results are displayed in Fig. \ref{fig:L1}. Here the resource allocation is reduced to only 23 edges while achieving the same risk as the results in Figs \ref{fig:P1Ours} and \ref{fig:Preciado008}, which allocated to 89 and 1109 edges, respectively.

\begin{figure}
  \centering
    \def\svgwidth{0.475\textwidth}
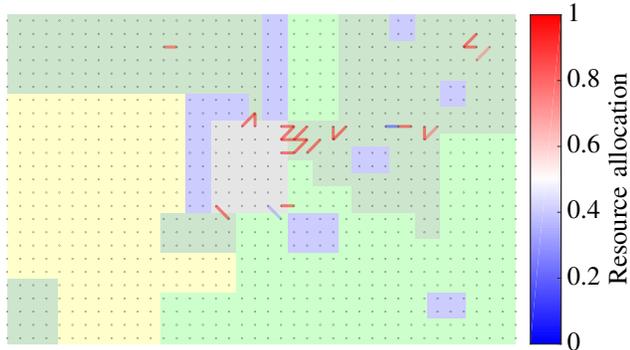
    \caption{Proposed resource allocation after using reweighted $\ell_{1}$ minimisation on Fig. \ref{fig:P1Ours}.}
    \label{fig:L1}
\end{figure}





%
%
%


%
%
%

\bibliographystyle{IEEEtran}
\bibliography{IEEEabrv,ACFR06022020}

\end{document}